\documentclass[a4paper,fleqn]{cas-sc}

\usepackage[authoryear,longnamesfirst]{natbib}
\usepackage{amsthm}
\usepackage{algorithm}
\usepackage{algpseudocode}

\newtheorem{thm}{Theorem}
\newtheorem{cor}[thm]{Corollary}
\newtheorem{prop}[thm]{Proposition}

\begin{document}
\let\WriteBookmarks\relax
\def\floatpagepagefraction{1}
\def\textpagefraction{.001}

\shorttitle{Gibbs Sampling for Bayesian GPMF}

\shortauthors{Sakaori and Abe}

\title [mode = title]{Gibbs Sampling for Bayesian Generalized Poisson Matrix Factorization}  



%

\author[1]{Fumitake Sakaori}[
    orcid=0000-0002-4383-4959
]

\cormark[1]


\ead{sfumitake001e@g.chuo-u.ac.jp}


\credit{Conceptualization, Methodology, Software, Formal analysis, Visualization, Writing - original draft}

\affiliation[1]{organization={Organization for the Establishment of the Faculty of Sports Data Science and Business, Chuo University},
            addressline={742-1 Higashinakano, Hachioji-shi}, 
            city={Tokyo},
            postcode={192-0393}, 
            country={Japan}}

\author[2]{Hiroyasu Abe}[
    orcid=0009-0004-1526-7143
]


\ead{hiro00012@gmail.com}


\credit{Conceptualization, Writing - review \& editing}

\affiliation[2]{organization={School of Pharmaceutical Sciences, Wakayama Medical University},
            addressline={25-1 Shichibancho, Wakayama-shi}, 
            city={Wakayama},
            postcode={640-8156}, 
            country={Japan}}

\cortext[1]{Corresponding author}



\begin{abstract}
Generalized Poisson matrix factorization (GPMF) is a matrix factorization method for count data with overdispersion based on the generalized Poisson distribution. While GPMF provides point estimates of the model parameters through maximum likelihood estimation, it does not quantify estimation uncertainty. In this paper, we propose a Bayesian extension of GPMF, referred to as Bayesian GPMF, and develop a Gibbs sampler for posterior inference. The proposed method is based on a compound Poisson representation of the generalized Poisson distribution, which introduces latent variables and yields closed-form full conditional posterior distributions for all model parameters. To efficiently sample the overdispersion parameter, we derive an infinite mixture representation of the exponentially tilted beta (EBeta) distribution and develop a finite approximation based on this representation. Simulation studies demonstrate that the proposed Bayesian GPMF achieves smaller mean squared errors than the conventional GPMF while providing credible intervals with empirical coverage probabilities close to the nominal level. Furthermore, the proposed finite approximation attains estimation accuracy comparable to Sampling/Importance Resampling (SIR) while requiring less computation. An application to football event data further illustrates the usefulness of Bayesian GPMF for extracting interpretable latent structures and quantifying estimation uncertainty. These results demonstrate that the proposed framework provides an effective Bayesian approach to generalized Poisson matrix factorization.
\end{abstract}




\maketitle

\section{Introduction}\label{sec:introduction}

Non-negative Matrix Factorization (NMF) is a matrix factorization technique that approximates a non-negative data matrix by the product of two non-negative matrices. Owing to the non-negativity constraint, the resulting decomposition provides an additive representation, making NMF highly interpretable. Consequently, NMF has been widely applied in various fields such as image processing, text mining, and recommender systems.

NMF can be formulated as an optimization problem that minimizes a divergence between an observed data matrix and its low-rank approximation. For count data, \citet{Lee1999} proposed an NMF model based on the Kullback--Leibler (KL) divergence. This formulation is equivalent to maximum likelihood estimation under the assumption that each observation independently follows a Poisson distribution, thereby providing a probabilistic interpretation of NMF.

The Poisson distribution assumes equality of the mean and variance, making it unsuitable for modeling overdispersed count data frequently encountered in practice. To address this issue, \citet{Consul01111973} introduced the Generalized Poisson Distribution (GPD), which is capable of modeling both overdispersion and underdispersion. Compared with the negative binomial distribution, the generalized Poisson distribution provides an alternative mechanism for modeling overdispersion, while exhibiting less zero inflation and heavier tails \citep{Joe2005}. Based on this distribution, \citet{Ohashi2026} proposed Generalized Poisson Matrix Factorization (GPMF), which extends Poisson matrix factorization by incorporating an additional dispersion parameter. As a result, GPMF provides greater flexibility for modeling count data than conventional Poisson matrix factorization.

In this paper, we extend GPMF to a Bayesian framework and develop a Gibbs sampling algorithm for posterior inference. By introducing latent variables based on a compound Poisson representation of the generalized Poisson distribution, all full conditional distributions can be derived in closed form, leading to an efficient Gibbs sampler. Bayesian inference also enables uncertainty quantification through posterior distributions of the model parameters.

The remainder of this paper is organized as follows. Section \ref{sec:background} reviews GPMF. Section \ref{sec:bgpmf} presents the proposed Bayesian GPMF model and derives the Gibbs sampling algorithm. Section \ref{sec:simulation} reports simulation studies, and Section \ref{sec:example} demonstrates the proposed method through an analysis of real data.

\section{Background}\label{sec:background}

Generalized Poisson Matrix Factorization (GPMF) combines nonnegative matrix factorization with the generalized Poisson distribution. This section first reviews these two components and then introduces GPMF.

\subsection{Nonnegative Matrix Factorization}\label{sec:nmf}

Nonnegative Matrix Factorization (NMF) is a technique for extracting latent structures from nonnegative data by approximating an observed nonnegative matrix with the product of two low-rank nonnegative matrices. Let $Y\in\mathbb{R}_{\ge0}^{I\times J}$ denote an observed data matrix. NMF approximates $Y$ by
\[
Y\simeq S=WH^\top,
\]
where $W\in\mathbb{R}_{\ge0}^{I\times K}$ and
$H\in\mathbb{R}_{\ge0}^{J\times K}$ are called the basis matrix and the coefficient matrix, respectively. Each element of $Y$ is approximated by
\[
y_{ij}\simeq s_{ij}
=
\sum_{k=1}^K w_{ik}h_{jk},
\]
where $K\ll\min(I,J)$ denotes the number of latent factors.

NMF can be interpreted within the framework of maximum likelihood estimation by assuming that each observation independently follows a probability distribution
\[
y_{ij}\sim F(s_{ij},\theta),
\]
where $F(s_{ij},\theta)$ denotes a probability distribution with mean $s_{ij}$ and nuisance parameter $\theta$. Under this formulation, minimizing the negative log-likelihood is, in many cases, equivalent to minimizing a divergence between the observed matrix and its low-rank approximation.

One of the most widely used models is Poisson Matrix Factorization (PMF) \citep{Lee1999}, which assumes
\[
y_{ij}\sim\mathrm{Poisson}(s_{ij}).
\]
Maximum likelihood estimation under this model is equivalent to minimizing the Kullback--Leibler divergence between the observed matrix $Y$ and its approximation $WH^\top$. However, because the Poisson distribution assumes equality of the mean and variance, it is not suitable for modeling overdispersed count data. Various extensions have therefore been proposed to overcome this limitation.

One such extension is Negative Binomial Matrix Factorization (NBMF) \citep{Gouvert2018,Gouvert2020}, which assumes
\[
y_{ij}
\sim
\mathrm{NB}
\left(
\alpha,
\frac{s_{ij}}{s_{ij}+\alpha}
\right).
\]
Since the negative binomial distribution can be represented as a gamma mixture of Poisson distributions, it naturally accommodates overdispersion. However, it has been pointed out that the distribution tends to produce excessive zeros when the degree of overdispersion is large \citep{MINAMI2007210}. In addition, estimating the dispersion parameter $\alpha$ is not straightforward, and it is often fixed in advance or selected from a set of candidate values.

\subsection{Generalized Poisson Distribution}\label{sec:gpd}

While NBMF models overdispersed count data using the negative binomial distribution, \citet{Ohashi2026} proposed a matrix factorization model based on the generalized Poisson distribution instead. We therefore briefly review the generalized Poisson distribution before introducing GPMF.

To relax the restrictive assumption of equal mean and variance in the Poisson distribution, \citet{Consul01111973} introduced the Generalized Poisson Distribution (GPD). A random variable following the generalized Poisson distribution, denoted by $\mathrm{GP}(\eta,\xi)$, has probability mass function
\begin{align}
p(y)=
\begin{cases}
\dfrac{\eta(\eta+\xi y)^{y-1}\exp\{-(\eta+\xi y)\}}{y!},
& y\ge0,\\
0,
& \mbox{for } y>m \mbox{ if } \xi<0.
\end{cases}
\end{align}
The parameter space is given by
\[
\eta>0,\qquad
\max(-1,-\eta/m)<\xi<1,
\]
where, for $\xi<0$, $m~(\ge4)$ denotes the largest integer satisfying
$\eta+\xi m>0$.
Since the probability mass function is defined as zero for all
$y>m$ when $\xi<0$, the total probability becomes less than one unless the distribution is treated as a truncated generalized Poisson distribution. In this paper, we therefore restrict our attention to the case $0<\xi<1$, where the generalized Poisson distribution serves as an overdispersion model.
Thus,
\[
V(Y)>E(Y)
\]
holds for $\xi>0$, implying that the generalized Poisson distribution naturally accommodates overdispersion.

Although both the negative binomial and generalized Poisson distributions can model overdispersion, their distributional characteristics differ. \citet{MINAMI2007210} pointed out that the negative binomial distribution tends to generate excessive zeros under strong overdispersion. Moreover, under matched mean and variance, \citet{Joe2005} showed that the negative binomial distribution exhibits greater zero inflation, whereas \citet{Ohashi2026} demonstrated that the generalized Poisson distribution has higher kurtosis, resulting in heavier tails.

\subsection{Generalized Poisson Matrix Factorization}\label{sec:gpmf}

Using the generalized Poisson distribution described in Section \ref{sec:gpd}, \citet{Ohashi2026} proposed Generalized Poisson Matrix Factorization (GPMF).

Let $Y\in\mathbb{R}_{\ge0}^{I\times J}$ denote an observed matrix. GPMF approximates $Y$ by the product of a basis matrix
$W\in\mathbb{R}_{\ge0}^{I\times K}$ and a coefficient matrix
$H\in\mathbb{R}_{\ge0}^{J\times K}$, namely
$Y\simeq WH^\top$, by assuming
\[
y_{ij}
\sim
\mathrm{GP}
\left(
\frac{s_{ij}}{1+\theta_i},
\frac{\theta_i}{1+\theta_i}
\right),
\]
where
\[
s_{ij}
=
\sum_{k=1}^K
w_{ik}h_{jk}.
\]
The corresponding mean and variance are
$s_{ij}$ and
$s_{ij}(1+\theta_i)^2$, respectively.

The corresponding negative log-likelihood is given by
\begin{align*}
-\ell(Y|W,H,\boldsymbol{\theta}) 
  &= - \sum_{i,j: y_{ij}\geq 1} \left\{
      \log s_{ij} + (y_{ij}-1)\log (s_{ij} +\theta_i y_{ij})
    \right\}\\
    & \hspace*{2em} +\sum_{i=1}^I \sum_{j=1}^J y_{ij} \log(1+\theta_i)
       +\sum_{i=1}^I \sum_{j=1}^J \frac{s_{ij}+\theta_i y_{ij}}{1+\theta_i}
       +\mathrm{const.},
\end{align*}
where
$\boldsymbol{\theta}
=
(\theta_1,\ldots,\theta_I)^\top$.

\citet{Ohashi2026} proposed an MM algorithm based on an auxiliary function for iteratively minimizing this objective function. Although the resulting algorithm provides point estimates of the model parameters, it does not quantify the uncertainty associated with the estimation.

\section{Bayesian GPMF}\label{sec:bgpmf}

GPMF has been formulated within the framework of maximum likelihood estimation, which provides point estimates of the model parameters. To quantify the uncertainty associated with parameter estimation, we extend GPMF to a Bayesian framework, in which inference is based on posterior distributions. Specifically, we formulate GPMF as a hierarchical Bayesian model and develop an inference method based on Gibbs sampling. This formulation is based on the compound Poisson representation of the generalized Poisson distribution. By introducing latent variables through this representation, conjugate prior distributions can be assigned to the elements of the factor matrices, enabling efficient sampling from the full conditional posterior distributions.

\subsection{Compound Poisson Representation of the Generalized Poisson Distribution}\label{sec:compoundpoisson}

The generalized Poisson distribution can be represented as a compound Poisson distribution. \citet{Finner2015} showed the following result.

\begin{thm}\label{thm:borel}
Let $Y_k~(k=1,\ldots,N)$ be independent random variables following the Borel distribution with parameter $\mu\in(0,1]$, and let $N$ be a random variable following the Poisson distribution with parameter $\eta>0$. Then,
\[
Z=\sum_{k=1}^{N}Y_k
\]
follows the generalized Poisson distribution $\mathrm{GP}(\eta,\mu)$.
\end{thm}

Here, the probability mass function of the Borel distribution $\mathrm{Borel}(\mu)$ is given by
\[
p(y)=
\dfrac{(\mu y)^{y-1}}{y!}
\exp(-\mu y),
\qquad
y=1,2,\ldots.
\]
Since $N$ follows a Poisson distribution, $Z$ follows a compound Poisson distribution. Thus, the generalized Poisson distribution can be interpreted as a compound Poisson distribution. When $N=0$, we define $Z=0$ with probability one.

Furthermore, the sum of random variables following the Borel distribution,
\[
Z=\sum_{k=1}^{N}Y_k,
\]
follows the Borel--Tanner distribution $\mathrm{BT}(N,\mu)$, whose probability mass function is given by
\[
p(z)
=
\dfrac{N}{z}
\dfrac{(\mu z)^{z-N}}{(z-N)!}
\exp(-\mu z),
\qquad
z=N,N+1,\ldots.
\]
When $N=0$, we define $Z=0$ with probability one. Using this representation, the following corollary follows immediately from Theorem \ref{thm:borel}.

\begin{cor}\label{cor:BT}
Let
\[
N\sim\mathrm{Poisson}(\eta),
\]
and suppose that the conditional distribution of $Y$ given $N$ is
\[
Y\mid N
\sim
\mathrm{BT}(N,\mu).
\]
Then, the marginal distribution of $Y$ is the generalized Poisson distribution
\[
Y\sim\mathrm{GP}(\eta,\mu).
\]
\end{cor}

In this paper, we formulate GPMF as a hierarchical Bayesian model based on this hierarchical representation.

\subsection{Bayesian GPMF Model}\label{sec:bgpmf_model}

Based on the hierarchical representation of the generalized Poisson distribution described in Section \ref{sec:compoundpoisson}, we define the hierarchical model for BGPMF as
\begin{align*}
  n_{ijk}\mid w_{ik},h_{jk},\xi_i
  &\sim
  \mathrm{Poisson}\bigl((1-\xi_i)w_{ik}h_{jk}\bigr), \\
  y_{ij}\mid n_{ij},\xi_i
  &\sim
  \mathrm{BT}(n_{ij},\xi_i), \\
  n_{ij}
  &=\sum_{k=1}^{K}n_{ijk}.
\end{align*}
The following result holds for this hierarchical model.

\begin{prop}
Under the hierarchical model defined above, the marginal distribution of $y_{ij}$ is
\[
y_{ij}
\sim
\mathrm{GP}\bigl((1-\xi_i)s_{ij},\xi_i\bigr).
\]
Furthermore,
\[
E(y_{ij})=s_{ij},
\qquad
V(y_{ij})
=
\frac{s_{ij}}{(1-\xi_i)^2}.
\]
\end{prop}

\begin{proof}
This follows immediately from Corollary \ref{cor:BT}.
\end{proof}

Using this hierarchical representation, we retain the same marginal distribution as in GPMF while facilitating the derivation of the Gibbs sampler described later.

Furthermore, given $w_{ik}$, $h_{jk}$, and $\xi_i$, the Poisson splitting property implies that the conditionally independent representation
\[
n_{ijk}\mid w_{ik},h_{jk},\xi_i
\sim
\mathrm{Poisson}\bigl((1-\xi_i)w_{ik}h_{jk}\bigr),
\qquad k=1,\ldots,K,
\]
is equivalent to
\begin{align*}
  n_{ij}\mid s_{ij},\xi_i
  &\sim
  \mathrm{Poisson}\bigl((1-\xi_i)s_{ij}\bigr),\\
  (n_{ij1},\ldots,n_{ijK})\mid n_{ij},W,H
  &\sim
  \mathrm{Multinomial}
  \left(n_{ij},\boldsymbol{\pi}_{ij}\right),\\
  \boldsymbol{\pi}_{ij}
  &=
  \left(
  \frac{w_{i1}h_{j1}}{s_{ij}},
  \ldots,
  \frac{w_{iK}h_{jK}}{s_{ij}}
  \right).
\end{align*}
In what follows, we use these two equivalent representations as appropriate for the derivations and algorithms.

Following the Bayesian formulation of NBMF \citep{Gouvert2018}, we assign independent gamma prior distributions to the elements of $W$ and $H$:
\begin{align*}
	w_{ik} & \sim \mathrm{Gamma}(\alpha_w, \beta_w),\\
	h_{jk} & \sim \mathrm{Gamma}(\alpha_h, \beta_h).
\end{align*}
For the prior distribution of $\boldsymbol{\xi}$, we assume that the elements $\xi_i$ are independent and have probability density function
\begin{align*}
	p(\xi_i)\propto \xi_i^{\alpha_\xi-1} (1-\xi_i)^{\beta_{\xi}-1}
	\exp\{-\gamma_{\xi} \xi_i \}.
\end{align*}
We denote this distribution by
\[
\mathrm{EBeta}(\alpha_\xi,\beta_\xi,\gamma_\xi).
\]
This distribution is obtained by applying the Esscher transform to a beta distribution. The Esscher transform, also known as exponential tilting, transforms a probability density function $p(x)$ into a new probability density function of the form
\[
p_{\gamma}(x)
=
\dfrac{e^{-\gamma x}}{E(e^{-\gamma x})}p(x).
\]
When $\gamma=0$, the distribution reduces to the ordinary beta distribution
$\mathrm{Beta}(\alpha_{\xi},\beta_{\xi})$.
Hereafter, we refer to this distribution as the exponentially tilted beta distribution.

Thus, treating $N=(n_{ijk})$ as latent variables, the joint posterior distribution of $W,H,\boldsymbol{\xi},N$ is given by
\begin{align*}
p(W,H,\boldsymbol{\xi},N|Y)
 & \propto p(Y|N,\boldsymbol{\xi})p(N|W,H,\boldsymbol{\xi})p(W)p(H)p(\boldsymbol{\xi})\\ 
 & \propto 
   \prod_{i,j:y_{ij}\ge 1} \dfrac{n_{ij}}{y_{ij}}
   \dfrac{\exp\{-\xi_i y_{ij} \}(\xi_i y_{ij})^{y_{ij}-n_{ij}}}
   {(y_{ij}-n_{ij})!}\\
 & \hspace*{2em}\cdot
   \prod_{i,j,k}
   \dfrac{\{(1-\xi_i)w_{ik}h_{jk}\}^{n_{ijk}}
   \exp\{-(1-\xi_i)w_{ik}h_{jk}\}}
   {n_{ijk}!}\\
 & \hspace*{2em}\cdot
   \prod_{i,k}w_{ik}^{\alpha_w-1}\exp\{-\beta_w w_{ik}\}
   \prod_{j,k}h_{jk}^{\alpha_h-1}\exp\{-\beta_h h_{jk}\}\\
 & \hspace*{2em}\cdot
   \prod_i \xi_i^{\alpha_{\xi}-1}(1-\xi_i)^{\beta_{\xi}-1}
   \exp\{-\gamma_{\xi}\xi_i\}.
\end{align*}

It is difficult to analytically derive the marginal posterior distributions from this joint posterior distribution. Therefore, in Section \ref{sec:gibbs}, we derive the full conditional posterior distributions of the model parameters and latent variables and construct a Gibbs sampler.

\subsection{Gibbs Sampler}\label{sec:gibbs}

We construct a Gibbs sampler by deriving the full conditional posterior distributions of the model parameters and latent variables. In what follows, ``$-$'' denotes all variables other than the variable of interest, together with the observed data.

First, the full conditional posterior distributions of $W$, $H$, and $\boldsymbol{\xi}$ are obtained from conjugacy as follows.

\begin{prop}
The full conditional posterior distributions of $W$, $H$, and $\boldsymbol{\xi}$ are
\begin{align*}
w_{ik}\mid -
&\sim
\mathrm{Gamma}
\left(
\alpha_w+\sum_j n_{ijk},
~
\beta_w+(1-\xi_i)\sum_j h_{jk}
\right),\\
h_{jk}\mid -
&\sim
\mathrm{Gamma}
\left(
\alpha_h+\sum_i n_{ijk},
~
\beta_h+\sum_i(1-\xi_i)w_{ik}
\right),\\
\xi_i\mid -
&\sim
\mathrm{EBeta}
\left(
\alpha_\xi+\sum_j(y_{ij}-n_{ij}),
~
\beta_\xi+\sum_j n_{ij},
~
\gamma_\xi+\sum_j(y_{ij}-s_{ij})
\right).
\end{align*}
\end{prop}

\begin{proof}
The result follows directly by combining the likelihood and prior distributions and using conjugacy.
\end{proof}

Next, we consider the full conditional posterior distribution of the latent variable $n_{ij}$. When $y_{ij}\ge 1$, $n_{ij}$ takes values in $1,\ldots,y_{ij}$, and the full conditional posterior distribution of $n_{ij}-1$ is a binomial distribution.

\begin{prop}
When $y_{ij}\ge1$,
\[
n_{ij}-1\mid -
\sim
\mathrm{Binomial}
\left(
y_{ij}-1,
\frac{(1-\xi_i)s_{ij}}
{(1-\xi_i)s_{ij}+\xi_i y_{ij}}
\right).
\]
When $y_{ij}=0$, $n_{ij}=0$ with probability one.
\end{prop}

\begin{proof}
The full conditional posterior distribution of $n_{ij}$ satisfies
\[
p(n_{ij}\mid-)
\propto
\frac{\{(1-\xi_i)s_{ij}\}^{n_{ij}}
(\xi_i y_{ij})^{y_{ij}-n_{ij}}}
{(y_{ij}-n_{ij})!(n_{ij}-1)!}.
\]
Letting $m_{ij}=n_{ij}-1$, we obtain
\[
p(m_{ij}\mid-)
\propto
\frac{\{(1-\xi_i)s_{ij}\}^{m_{ij}}
(\xi_i y_{ij})^{y_{ij}-1-m_{ij}}}
{(y_{ij}-1-m_{ij})!m_{ij}!},
\qquad
m_{ij}=0,\ldots,y_{ij}-1.
\]
Therefore,
\[
m_{ij}\mid-
\sim
\mathrm{Binomial}
\left(
y_{ij}-1,
\frac{(1-\xi_i)s_{ij}}
{(1-\xi_i)s_{ij}+\xi_i y_{ij}}
\right),
\]
which proves the result.
\end{proof}

After updating $n_{ij}$, its value needs to be allocated to
$n_{ij1},\ldots,n_{ijK}$ corresponding to the $K$ latent factors.
By the Poisson splitting property, their conditional distribution is a multinomial distribution.

\begin{prop}
Given $n_{ij}$,
\[
(n_{ij1},\ldots,n_{ijK})\mid -,n_{ij}
\sim
\mathrm{Multinomial}
\left(
n_{ij},
\boldsymbol{\pi}_{ij}
\right),
\]
where
\[
\boldsymbol{\pi}_{ij}
=
\left(
\frac{w_{i1}h_{j1}}{s_{ij}},
\ldots,
\frac{w_{iK}h_{jK}}{s_{ij}}
\right).
\]
\end{prop}

\begin{proof}
This follows immediately from the Poisson splitting property.
\end{proof}

From the above results, the Gibbs sampler updates the variables in each iteration in the following order:
\[
n_{ij}
\rightarrow
(n_{ij1},\ldots,n_{ijK})
\rightarrow
W
\rightarrow
H
\rightarrow
\boldsymbol{\xi}.
\]
Algorithm \ref{alg:gibbs} summarizes the Gibbs sampling procedure.

A practical issue in implementing the Gibbs sampler is sampling from the EBeta distribution when updating $\xi_i$. A simple approach is Sampling/Importance Resampling (SIR), in which a large number of samples are generated from the beta distribution
\[
\mathrm{Beta}
\left(
\alpha_\xi+\sum_j(y_{ij}-n_{ij}),
~
\beta_\xi+\sum_j n_{ij}
\right)
\]
and one of them is selected with probability proportional to the weight
\[
\exp\left\{
-\left(
\gamma_\xi+\sum_j(y_{ij}-s_{ij})
\right)\xi_i
\right\}.
\]

However, this approach requires generating a large number of candidate samples and is therefore not necessarily computationally efficient. In Section \ref{sec:ebeta}, we develop a more efficient sampling method by representing the EBeta distribution as an infinite mixture of beta distributions.

\begin{algorithm}
\caption{Gibbs sampler for the BGPMF model}
\label{alg:gibbs}
\begin{algorithmic}[1]

\Require Observed matrix $Y$; hyperparameters
$\alpha_w,\beta_w,\alpha_h,\beta_h,\alpha_\xi,\beta_\xi,\gamma_\xi$;
number of iterations $T$

\Ensure Posterior samples of $W,H,\boldsymbol{\xi},N$

\State Initialize $W,H,\boldsymbol{\xi},N$

\For{$t=1,\ldots,T$}

  \State Compute
  $s_{ij}\gets\sum_{k=1}^{K}w_{ik}h_{jk}$
  for all $(i,j)$

  \For{each $(i,j)$}

    \If{$y_{ij}=0$}
      \State Set $n_{ij}\gets0$
    \Else
      \State Sample
      $n_{ij}\sim
      1+\mathrm{Binomial}\!\left(
      y_{ij}-1,
      \dfrac{(1-\xi_i)s_{ij}}
      {(1-\xi_i)s_{ij}+\xi_i y_{ij}}
      \right)$
    \EndIf

    \State Set
    $\pi_{ijk}\gets w_{ik}h_{jk}/s_{ij}$
    for $k=1,\ldots,K$

    \State Sample
    $(n_{ij1},\ldots,n_{ijK})
    \sim
    \mathrm{Multinomial}
    \left(n_{ij},\boldsymbol{\pi}_{ij}\right)$

  \EndFor

  \For{each $(i,k)$}

    \State Sample
    $w_{ik}\sim
    \mathrm{Gamma}\!\left(
    \alpha_w+\sum_j n_{ijk},
    \beta_w+(1-\xi_i)\sum_j h_{jk}
    \right)$

  \EndFor

  \For{each $(j,k)$}

    \State Sample
    $h_{jk}\sim
    \mathrm{Gamma}\!\left(
    \alpha_h+\sum_i n_{ijk},
    \beta_h+\sum_i(1-\xi_i)w_{ik}
    \right)$

  \EndFor

  \State Recompute
  $s_{ij}\gets\sum_{k=1}^{K}w_{ik}h_{jk}$
  for all $(i,j)$

  \For{each $i$}

    \State Set
    $a_i\gets
    \alpha_\xi+\sum_j(y_{ij}-n_{ij})$

    \State Set
    $b_i\gets
    \beta_\xi+\sum_j n_{ij}$

    \State Set
    $g_i\gets
    \gamma_\xi+\sum_j(y_{ij}-s_{ij})$

    \State Sample
    $\xi_i\sim
    \mathrm{EBeta}(a_i,b_i,g_i)$
    using Algorithm \ref{alg:ebeta}

  \EndFor

  \State Store $W,H,\boldsymbol{\xi},N$

\EndFor

\end{algorithmic}
\end{algorithm}

\subsection{Infinite Mixture Representation of the EBeta Distribution}\label{sec:ebeta}

In Section \ref{sec:gibbs}, we noted that SIR can be used to generate samples from the EBeta distribution. However, this method requires generating a large number of random variables from a beta distribution to obtain a single sample and is therefore computationally inefficient. In this section, we first represent the EBeta distribution as an infinite mixture of beta distributions. We then describe a sampling method based on this infinite mixture representation and introduce the finite approximation used in this study.

Including the normalizing constant, the probability density function of
$\mathrm{EBeta}(\alpha,\beta,\gamma)$ is
\[
p(\xi)=
\dfrac{1}{B(\alpha,\beta){}_1F_1(\alpha;\alpha+\beta;-\gamma)}
\xi^{\alpha-1}(1-\xi)^{\beta-1}\exp\{-\gamma\xi\},
\]
where ${}_1F_1$ denotes the confluent hypergeometric function. In what follows, we derive an infinite mixture representation of the EBeta distribution by expanding this probability density function.

First, consider the case $\gamma\geq0$. The exponential term can be expanded as a power series in $(1-\xi)$:
\[
\exp\{-\gamma\xi\}
=
\exp\{-\gamma\}\exp\{\gamma(1-\xi)\}
=
\exp\{-\gamma\}
\sum_{m=0}^{\infty}
\dfrac{\gamma^m}{m!}(1-\xi)^m.
\]
Therefore,
\[
p(\xi)
\propto
\sum_{m=0}^{\infty}
\dfrac{\gamma^m}{m!}
\xi^{\alpha-1}(1-\xi)^{\beta+m-1}.
\]
Thus, the EBeta distribution can be represented as an infinite mixture of beta distributions
$\mathrm{Beta}(\alpha,\beta+m)$.

The mixing probability $p_m$ corresponding to each beta distribution is given by
\begin{align*}
p_m
&=
\dfrac{
\frac{\gamma^m}{m!}B(\alpha,\beta+m)
}{
\displaystyle\sum_{l=0}^{\infty}
\frac{\gamma^l}{l!}B(\alpha,\beta+l)
},
\qquad
m=0,1,\ldots.
\end{align*}
Using the Pochhammer symbol
\[
(\beta)_m
=
\beta(\beta+1)\cdots(\beta+m-1),
\]
we can write
\[
B(\alpha,\beta+m)
=
\dfrac{\Gamma(\alpha)\Gamma(\beta+m)}
{\Gamma(\alpha+\beta+m)}
=
B(\alpha,\beta)
\dfrac{(\beta)_m}{(\alpha+\beta)_m}.
\]
Hence,
\begin{align*}
p_m
&=
\dfrac{1}
{{}_1F_1(\beta;\alpha+\beta;\gamma)}
w_m,
\qquad
w_m
=
\dfrac{(\beta)_m}{(\alpha+\beta)_m}
\dfrac{\gamma^m}{m!}.
\end{align*}
The weights $w_m$ can be computed recursively using
\begin{align}
w_{m+1}
=
\dfrac{(\beta+m)\gamma}
{(\alpha+\beta+m)(m+1)}
w_m.
\label{eq:positive-gamma}
\end{align}
By numerically evaluating the confluent hypergeometric function
${}_1F_1(\beta;\alpha+\beta;\gamma)$, the weights $w_m$ can be computed recursively and normalized to obtain $p_m$. The mixture component index $m$ can then be generated by the inverse transform method by comparing a uniform random variable with the cumulative probabilities of $p_m$.

\begin{algorithm}
\caption{Sampling from the finite approximation to
$\mathrm{EBeta}(\alpha,\beta,\gamma)$}
\label{alg:ebeta}
\begin{algorithmic}[1]

\Require Parameters $\alpha,\beta,\gamma$;
truncation level $M$

\Ensure A sample $\xi$ from the finite approximation to
$\mathrm{EBeta}(\alpha,\beta,\gamma)$

\If{$\gamma\ge0$}

\State Set $w_0\gets1$

\For{$m=0,\ldots,M-1$}
    \State Compute
    \[
    w_{m+1}\gets
    \frac{(\beta+m)\gamma}
    {(\alpha+\beta+m)(m+1)}
    w_m
    \]
\EndFor

\State Compute
\[
\tilde p_m
\gets
\frac{w_m}{\sum_{l=0}^{M}w_l},
\qquad m=0,\ldots,M
\]

\State Sample
$R\sim
\mathrm{Categorical}
(\tilde p_0,\ldots,\tilde p_M)$

\State Generate
$\xi\sim
\mathrm{Beta}
(\alpha,\beta+R)$

\Else

\State Set $w_0\gets1$

\For{$m=0,\ldots,M-1$}
    \State Compute
    \[
    w_{m+1}\gets
    \frac{(\alpha+m)|\gamma|}
    {(\alpha+\beta+m)(m+1)}
    w_m
    \]
\EndFor

\State Compute
\[
\tilde p_m
\gets
\frac{w_m}{\sum_{l=0}^{M}w_l},
\qquad m=0,\ldots,M
\]

\State Sample
$R\sim
\mathrm{Categorical}
(\tilde p_0,\ldots,\tilde p_M)$

\State Generate
$\xi\sim
\mathrm{Beta}
(\alpha+R,\beta)$

\EndIf

\end{algorithmic}
\end{algorithm}

In this study, for simplicity of implementation, we instead use a sufficiently large truncation level $M$ and approximate $p_m$ by
\[
\tilde{p}_m
=
\dfrac{w_m}
{\displaystyle\sum_{l=0}^M w_l},
\qquad
m=0,\ldots,M.
\]
The mixture component index $m$ is then generated using these approximate mixing probabilities.

Next, consider the case $\gamma<0$. In this case, the exponential term can be expanded as a power series in $\xi$:
\[
\exp\{-\gamma\xi\}
=
\exp\{|\gamma|\xi\}
=
\sum_{m=0}^{\infty}
\dfrac{|\gamma|^m}{m!}\xi^m.
\]
Thus,
\[
p(\xi)
\propto
\xi^{\alpha-1}(1-\xi)^{\beta-1}\exp\{-\gamma\xi\}
\propto
\sum_{m=0}^{\infty}
\dfrac{|\gamma|^m}{m!}
\xi^{\alpha+m-1}(1-\xi)^{\beta-1}.
\]
Therefore, the EBeta distribution can be represented as an infinite mixture of beta distributions
$\mathrm{Beta}(\alpha+m,\beta)$.

As in the case $\gamma\ge0$, the mixing probability $p_m$ and weight $w_m$ are given by
\begin{align}
p_m
&=
\dfrac{1}
{{}_1F_1(\alpha;\alpha+\beta;|\gamma|)}
w_m,\nonumber\\
w_m
&=
\dfrac{(\alpha)_m}{(\alpha+\beta)_m}
\dfrac{|\gamma|^m}{m!},\nonumber\\
w_{m+1}
&=
\dfrac{(\alpha+m)|\gamma|}
{(\alpha+\beta+m)(m+1)}
w_m.
\label{eq:negative-gamma}
\end{align}
In this case, as in the case $\gamma\ge0$, the mixture component index $m$ can be generated either by numerically evaluating the confluent hypergeometric function and applying the inverse transform method, or by using a finite approximation to $p_m$.

It follows that a sample from the EBeta distribution can be generated by first selecting a mixture component index $m$ according to the mixing probabilities $p_m$ and then generating a random variable from the corresponding beta distribution,
$\mathrm{Beta}(\alpha,\beta+m)$ when $\gamma\ge0$, or
$\mathrm{Beta}(\alpha+m,\beta)$ when $\gamma<0$.
In this study, we select the mixture component index $m$ using the finite approximation $\tilde{p}_m$ to the mixing probabilities $p_m$. Algorithm \ref{alg:ebeta} summarizes the finite approximation method used in this study.

We next examine the accuracy of the finite approximation used in this study. Preliminary numerical experiments showed that the approximation error was only weakly affected by the values of $\alpha$ and $\beta$, whereas it was primarily determined by the magnitude of the tilting parameter $\gamma$. Therefore, we fix $\alpha=\beta=5$ and evaluate the accuracy of the finite approximation for $\gamma=5$ and $\gamma=100$.

Table \ref{tab:precision} shows the differences between the approximated and true values of the mean, variance, and percentiles for various values of $M$. The true values of the mean and variance were obtained from theoretical expressions involving the confluent hypergeometric function ${}_1F_1$, whereas the true percentiles were obtained by numerically evaluating the quantile function of the EBeta distribution. Even for a large tilting parameter such as $\gamma=100$, the errors decreased to the level of machine precision when $M=200$. Therefore, we use $M=200$ in the numerical experiments that follow.

\begin{table}[tb]
\centering
\caption{Accuracy of the finite approximation (approximation $-$ true value)}
\label{tab:precision}

\begin{tabular}{ccccccc}\hline
$\gamma$ & $M$ & Mean error & Var error & $1\%$ error & $50\%$ error & $99\%$ error \\\hline
  & 10  & $1.409\times 10^{-4}$ & $-1.143\times 10^{-5}$ & $1.738\times 10^{-4}$ & $1.274\times 10^{-4}$ & $4.652\times 10^{-5}$ \\
  & 50  & $4.440\times 10^{-16}$ & $-1.665\times 10^{-16}$ & $-4.163\times 10^{-17}$ & $-3.331\times 10^{-16}$ & $-7.438\times 10^{-15}$ \\
5 & 100 & $4.440\times 10^{-16}$ & $-1.665\times 10^{-16}$ & $-4.163\times 10^{-17}$ & $-3.331\times 10^{-16}$ & $-7.438\times 10^{-15}$ \\
  & 200 & $4.440\times 10^{-16}$ & $-1.665\times 10^{-16}$ & $-4.163\times 10^{-17}$ & $-3.331\times 10^{-16}$ & $-7.438\times 10^{-15}$ \\
  & 500 & $4.440\times 10^{-16}$ & $-1.665\times 10^{-16}$ & $-4.163\times 10^{-17}$ & $-3.331\times 10^{-16}$ & $-7.438\times 10^{-15}$ \\\hline
  & 10 & $2.041\times 10^{-1}$ & $8.620\times 10^{-3}$ & $6.107\times 10^{-2}$ & $1.987\times 10^{-1}$ & $3.912\times 10^{-1}$ \\
  & 50 & $3.699\times 10^{-2}$ & $8.454\times 10^{-4}$ & $1.035\times 10^{-2}$ & $3.541\times 10^{-2}$ & $7.695\times 10^{-2}$ \\
100 & 100 & $2.143\times 10^{-3}$ & $2.887\times 10^{-5}$ & $7.026\times 10^{-4}$ & $2.138\times 10^{-3}$ & $3.621\times 10^{-3}$ \\ 
  & 200 & $1.596\times 10^{-16}$ & $-9.541\times 10^{-18}$ & $-5.204\times 10^{-18}$ & $-6.245\times 10^{-17}$ & $-3.039\times 10^{-15}$ \\
  & 500 & $1.596\times 10^{-16}$ & $-9.541\times 10^{-18}$ & $-5.204\times 10^{-18}$ & $-6.245\times 10^{-17}$ & $-3.039\times 10^{-15}$ \\\hline
\end{tabular}

\end{table}

\section{Simulation Studies}\label{sec:simulation}

In this section, we evaluate the estimation accuracy and computational efficiency of Bayesian GPMF through simulation studies. As a common setting, the dimensions of the data matrix were set to $I=50$ and $J=100$. The number of latent factors was set to $K=5$ and assumed to be known. The elements of the true basis matrix $W_0 \in \mathbb{R}^{I\times K}_{\geq 0}$ and coefficient matrix $H_0 \in \mathbb{R}^{J\times K}_{\geq 0}$ were independently generated from $\mathrm{Gamma}(1.5,1.5)$. We then computed $S_0=W_0H_0^{\top}$ and independently generated $y_{ij}$ from the generalized Poisson distribution
\[
y_{ij}
\sim
\mathrm{GP}\bigl((1-\xi_{0i})(s_0)_{ij},\xi_{0i}\bigr),
\]
where $(s_0)_{ij}$ denotes the $(i,j)$th element of $S_0$.

These settings are the same as those used in \citet{Ohashi2026}, except that \citet{Ohashi2026} used $\theta$ as the dispersion parameter, whereas we use $\xi$ in this paper. Each simulation experiment was repeated 100 times. Except for the experiment in Section \ref{sec:fam}, samples from the EBeta distribution were generated using the finite approximation method described in Section \ref{sec:ebeta}.

\subsection{Comparison with GPMF}\label{sec:mse}

We first compare Bayesian GPMF with GPMF. As evaluation measures, we use the mean squared errors (MSEs) of $W$, $H$, $S$, and $\xi$. Following \citet{Ohashi2026}, we conducted simulation experiments for both the case where $\xi_{0i}$ is the same for all $i$ and the case where $\xi_{0i}$ varies across $i$.

\citet{Ohashi2026} also considered $\theta_{0i}=0$, which corresponds to $\xi_{0i}=0$. In this paper, however, the support of the prior distribution for $\xi$ is $(0,1)$, and therefore we do not consider the boundary case $\xi_{0i}=0$. Instead, we include $\xi_{0i}=1/5$. The other values of $\xi_{0i}$ are $1/3$, $1/2$, $3/5$, and $2/3$, which correspond to $\theta_{0i}=0.5$, $1$, $1.5$, and $2$, respectively, in \citet{Ohashi2026}.

We also considered a heterogeneous setting in which the dispersion parameter varies across rows. In this setting, each of the five values above was assigned to 10 rows.

Considering the balance between computational cost and estimation accuracy, we ran the Gibbs sampler for 5,000 iterations with the first 1,000 iterations discarded as burn-in.

The results are shown in Table \ref{tab:mse}. For GPMF, we use the same results as those reported in \citet{Ohashi2026}, except for the case $\xi_{0i}=1/5$, for which we conducted an additional simulation experiment. As shown in Table \ref{tab:mse}, Bayesian GPMF yielded smaller MSEs than GPMF under all settings. In particular, the improvements in the estimation accuracy of the basis matrix $W$ and mean matrix $S$ were substantial, and the difference between Bayesian GPMF and GPMF tended to increase as the dispersion parameter $\xi$ increased. A similar tendency was observed under the heterogeneous setting, indicating that the advantage of Bayesian GPMF was maintained even when the dispersion parameter varied across rows.

\begin{table}[tb]
\centering
\caption{Average MSEs over 100 simulation runs.}
\label{tab:mse}
\begin{tabular}{lcccc}
\toprule
$\xi_{0i}$ & Parameter & \multicolumn{2}{c}{GPMF} & BGPMF \\
\cmidrule(lr){3-4}
 & & NNDSVD & Random &  \\
\midrule
\multirow{4}{*}{$1/5$}
& $W$   & 73.9 & 54.2 & 28.2 \\
& $H$   & 0.00219 & 0.00192 & 0.00150 \\
& $S$   & 1.70 & 1.64 & 1.05 \\
& $\xi$ & --- & --- & 0.00405 \\
\midrule
\multirow{4}{*}{$1/3$}
& $W$   & 110 & 90.1 & 42.6 \\
& $H$   & 0.00311 & 0.00288 & 0.00206 \\
& $S$   & 2.59 & 2.51 & 1.40 \\
& $\xi$ & --- & --- & 0.00437 \\
\midrule
\multirow{4}{*}{$1/2$}
& $W$   & 186 & 158 & 65.2 \\
& $H$   & 0.00444 & 0.00429 & 0.00283 \\
& $S$   & 4.59 & 4.25 & 1.94 \\
& $\xi$ & --- & --- & 0.00350 \\
\midrule
\multirow{4}{*}{$3/5$}
& $W$   & 264 & 211 & 75.7 \\
& $H$   & 0.00556 & 0.00526 & 0.00313 \\
& $S$   & 6.89 & 5.94 & 2.29 \\
& $\xi$ & --- & --- & 0.00311 \\
\midrule
\multirow{4}{*}{$2/3$}
& $W$   & 335 & 251 & 84.1 \\
& $H$   & 0.00652 & 0.00597 & 0.00323 \\
& $S$   & 9.95 & 7.61 & 2.66 \\
& $\xi$ & --- & --- & 0.00250 \\
\midrule
\multirow{4}{*}{Heterogeneous}
& $W$   & 286 & 153 & 57.3 \\
& $H$   & 0.00551 & 0.00409 & 0.00249 \\
& $S$   & 6.44 & 4.20 & 1.81 \\
& $\xi$ & --- & --- & 0.00346 \\
\bottomrule
\end{tabular}
\end{table}

\subsection{Coverage of Credible Intervals}\label{sec:cov}

Next, we examine the empirical coverage probabilities of the 95\% credible intervals obtained by Bayesian GPMF, that is, the proportions of credible intervals that contain the true values. To avoid the effects of the scale indeterminacy of $W$ and $H$ and the indeterminacy in the ordering of the latent factors, we constructed credible intervals for each element of $S=WH^{\top}$ and for each $\xi_i$. For each element, we calculated the proportion of the 100 simulation runs in which the credible interval contained the true value, and then averaged these proportions over all elements.

The results are shown in Table \ref{tab:coverage}. For both $S$ and $\xi_i$, the empirical coverage probabilities were generally close to the nominal level of 95\%. In particular, the coverage probabilities for the mean matrix $S$ were close to 95\% under all settings, indicating that the credible intervals obtained by Bayesian GPMF had appropriate coverage. For the dispersion parameter $\xi$, the coverage probabilities were slightly lower in some settings but were generally close to the nominal level.

\begin{table}[tb]
\centering
\caption{Empirical coverage probabilities of 95\% credible intervals for BGPMF.}
\label{tab:coverage}
\begin{tabular}{lcc}
\toprule
$\xi_{0i}$ & $S$ & $\xi$ \\
\midrule
$1/5$    & 0.954 & 0.953 \\
$1/3$    & 0.954 & 0.949 \\
$1/2$    & 0.959 & 0.944 \\
$3/5$    & 0.964 & 0.934 \\
$2/3$    & 0.966 & 0.947 \\
Heterogeneous & 0.958 & 0.948 \\
\bottomrule
\end{tabular}
\end{table}

\subsection{Behavior of the Gibbs Sampler}

To assess the convergence of the Gibbs sampler, we examine the effective sample size (ESS), trace plot, and autocorrelation function (ACF) using one of the 100 simulation runs as an example.

First, we calculated the ESSs of $\xi_i$ for two settings: 5,000 iterations with 1,000 burn-in iterations and 20,000 iterations with 4,000 burn-in iterations. Under the heterogeneous setting in Section \ref{sec:mse}, we calculated the mean, median, minimum, and maximum ESS for the rows corresponding to each value of $\xi_{0i}$. The results are shown in Table \ref{tab:ess}.

The ESS generally increased as the number of iterations increased. In particular, the rows with larger values of $\xi_{0i}$ showed better sampling efficiency and relatively large ESSs. For smaller values of $\xi_{0i}$, the ESSs were relatively small, but they improved substantially as the number of iterations increased.

\begin{table}[tb]
\centering
\caption{Effective sample sizes for $\xi_i$.}
\label{tab:ess}
\begin{tabular}{lllrrrr}
\toprule
Iterations & Burn-in & $\xi_{0i}$ & Mean ESS & Median ESS & Minimum ESS & Maximum ESS \\
\midrule
 & & $1/5$ & 177 & 124 & 49.7 & 494 \\
 & & $1/3$ & 202 & 221 & 46.5 & 418 \\
5000 & 1000 & $1/2$ & 478 & 473 & 183 & 672 \\
 & & $3/5$ & 541 & 548 & 370 & 674 \\
 & & $2/3$ & 545 & 543 & 487 & 615 \\
\midrule
 & & $1/5$ & 540 & 343 & 173 & 1434 \\
 & & $1/3$ & 800 & 697 & 219 & 2013 \\
20000 & 4000 & $1/2$ & 1644 & 1679 & 619 & 2402 \\
 & & $3/5$ & 2095 & 2270 & 1535 & 2600 \\
 & & $2/3$ & 2086 & 2120 & 1522 & 2360 \\
\bottomrule
\end{tabular}
\end{table}

For one of the rows corresponding to $\xi_{0i}=2/3$, the trace plot and ACF based on 20,000 iterations are shown in Figures \ref{fig:traceplot} and \ref{fig:acf}, respectively.

\begin{figure}[tb]
\centering
\includegraphics[width=14cm]{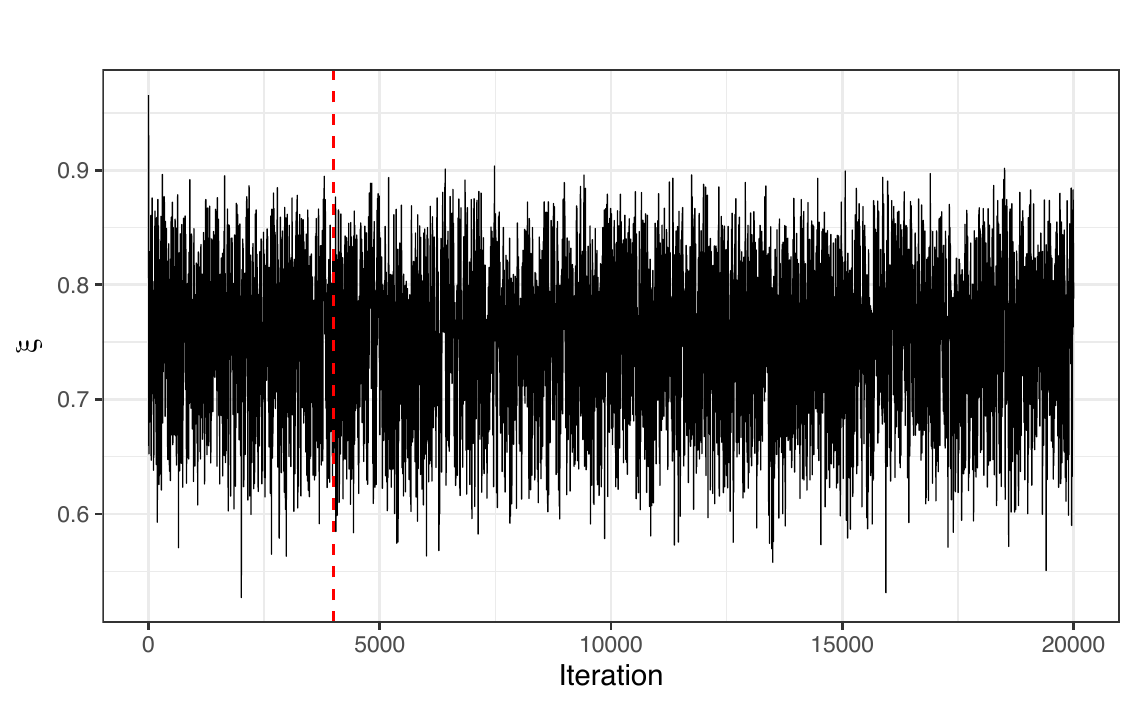}
\caption{Trace plot of $\xi_i$.}
\label{fig:traceplot}
\end{figure}

\begin{figure}[tb]
\centering
\includegraphics[width=14cm]{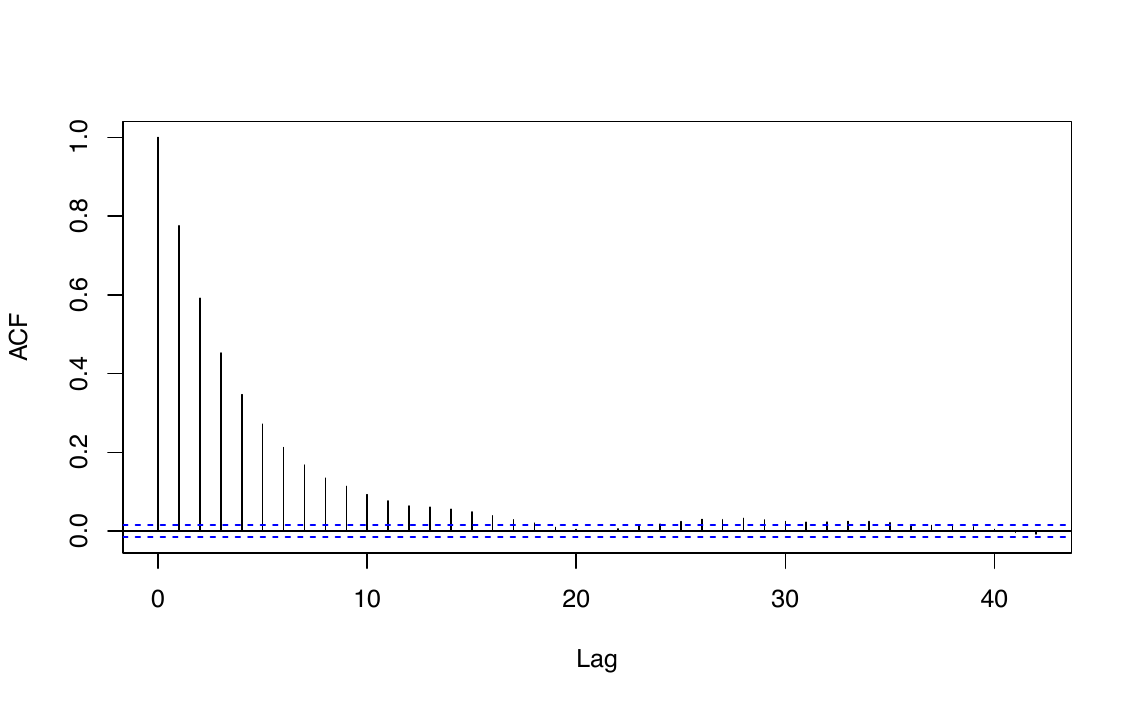}
\caption{Autocorrelation function of $\xi_i$.}
\label{fig:acf}
\end{figure}

The trace plot suggests that the displayed chain reached a stationary regime relatively quickly, with no apparent indication that a burn-in period longer than 1,000 iterations was needed for this chain. Although the ACF shows some autocorrelation at small lags, the ESS improved as the number of iterations increased. Moreover, as shown in Sections \ref{sec:mse} and \ref{sec:cov}, using 5,000 iterations yielded smaller MSEs than GPMF and empirical coverage probabilities generally close to the nominal level. Therefore, considering the balance between computational cost and estimation accuracy, we regarded 5,000 iterations with 1,000 burn-in iterations as a reasonable setting for the repeated simulation experiments.

\subsection{Comparison between SIR and the Finite Approximation Method}\label{sec:fam}

We compare the Sampling/Importance Resampling (SIR) method and the finite approximation method based on the infinite mixture representation for sampling from the EBeta distribution in terms of computation time, MSE, and coverage probability. For the SIR method, the numbers of candidate samples drawn from the beta proposal distribution were set to $L=100,200,500,$ and $1000$. To avoid numerical underflow and overflow, the maximum log-importance weight was subtracted from all log-importance weights before exponentiation. For the finite approximation method, the truncation levels were set to $M=100$ and $200$.

As shown in Table \ref{tab:comparison}, there were no substantial differences in the MSEs of the model parameters between the SIR and finite approximation methods. The coverage probabilities for $S$ were also close to the nominal level for both methods. For $\xi$, the coverage probability obtained by the SIR method tended to improve as the number of candidate samples $L$ increased, but it was still 0.941 even when $L=1000$. In contrast, the finite approximation method yielded coverage probabilities of 0.945 and 0.948 for $M=100$ and $M=200$, respectively, which were closer to the nominal level. Moreover, the computation times for the finite approximation method were shorter than that for the SIR method with $L=1000$. These results indicate that the finite approximation method is more computationally efficient when a comparable level of accuracy in interval estimation for $\xi$ is required.

\begin{table}[tb]
\centering
\caption{Comparison between the SIR and finite approximation methods.}
\label{tab:comparison}
\begin{tabular}{lrrrrrrrr}
\toprule
Method & Parameter & Time (s) & MSE$(W)$ & MSE$(H)$ & MSE$(S)$ & MSE$(\xi)$ & CP$(S)$ & CP$(\xi)$ \\
\midrule
& $L=100$  & 3227.1 & 60.1 & 0.00254 & 1.84 & 0.00362 & 0.956 & 0.925 \\
SIR & $L=200$  & 3439.3 & 56.5 & 0.00253 & 1.76 & 0.00352 & 0.959 & 0.934 \\
& $L=500$  & 4092.2 & 59.2 & 0.00256 & 1.80 & 0.00356 & 0.958 & 0.939 \\
& $L=1000$ & 6202.0 & 57.8 & 0.00248 & 1.81 & 0.00353 & 0.957 & 0.941 \\
\midrule
Finite approximation & $M=100$ & 4090.4 & 61.0 & 0.00251 & 1.85 & 0.00351 & 0.958 & 0.945 \\
& $M=200$ & 4706.2 & 57.3 & 0.00249 & 1.81 & 0.00346 & 0.958 & 0.948 \\
\bottomrule
\end{tabular}

\vspace{2mm}
\raggedright
\footnotesize
* Run time was measured on a MacBook Pro equipped with an Apple M5 processor and 16 GB of memory.
\end{table}

\section{Real Data Example}\label{sec:example}

As an application of the proposed Bayesian GPMF to real data, we analyze gegenpressing counts obtained from open football data provided by StatsBomb. We construct an observed matrix $Y$ containing the numbers of gegenpressing events for 20 teams over 38 matches and apply both GPMF and Bayesian GPMF to this matrix.

\subsection{Data and Preprocessing}

Gegenpressing, also known as counter-pressing, is a football tactic in which a team applies immediate and intense pressure after losing possession in an attempt to regain the ball. For this analysis, we used the open event data provided by StatsBomb for all matches in the 2015--2016 English Premier League season, in which individual events during each match are recorded.

We defined a gegenpressing event as a Pressure event occurring within five seconds after a team lost possession. We counted the number of gegenpressing events for each of the 20 teams in each matchweek and constructed a $20\times38$ gegenpressing matrix $Y$. The rows of $Y$ correspond to the teams, and the columns correspond to the 38 matchweeks. By applying matrix factorization to this matrix, we aim to identify differences in pressing styles among teams and temporal patterns over the course of the season.

\subsection{Results}

We applied Bayesian GPMF to the gegenpressing matrix $Y$. To obtain more stable posterior estimates, we set the truncation level of the finite approximation to $M=500$ and ran the Gibbs sampler for 10,000 iterations, with the first 2,000 iterations discarded as burn-in. We set $\alpha_w=\beta_w=\alpha_h=\beta_h=\alpha_\xi=\beta_\xi=1$ and $\gamma_\xi=0$, so that the prior distribution for $\xi_i$ reduces to the uniform distribution on $(0,1)$.

Figure \ref{fig:w_bgpmf} shows the posterior means and 95\% credible intervals for the basis matrix $W$ and the overdispersion parameter $\xi$. Although the values of the first basis vary across teams, they are relatively large for many teams, suggesting that this basis represents a basic level of pressing intensity common across the league. In contrast, the second and third bases show greater variation across teams and take values close to zero for most teams. These bases can therefore be interpreted as representing specific tactical characteristics. The widths of the credible intervals also vary across teams, particularly for the second and third bases, demonstrating that Bayesian GPMF allows us to assess the uncertainty in the estimates for individual teams.

Figure \ref{fig:h_bgpmf} shows the posterior means and 95\% credible intervals for the coefficient matrix $H$. The coefficients of the first basis remain relatively stable throughout the season, which is consistent with its interpretation as representing a basic level of pressing intensity common across the league. The second basis shows a pronounced peak in matchweek 25, whereas the third basis shows pronounced peaks in several matchweeks during the first half of the season and in matchweek 27. This suggests that these bases represent tactical characteristics that appeared during particular periods of the season. In particular, the coefficients of the second basis are relatively large in matchweeks 20, 25, 28, and 29, which correspond to congested fixture periods. Thus, this basis may reflect whether teams were able to maintain their pressing intensity during such periods.

Regarding the overdispersion parameter $\xi_i$, the largest values were observed for Watford, followed by Swansea City, Chelsea, and Tottenham Hotspur. These relatively large values indicate greater match-to-match variation in the number of gegenpressing events for these teams and may reflect changes in their tactical behavior. In contrast, Leicester City had a particularly small value of $\xi_i$, indicating relatively little variation in the number of gegenpressing events throughout the season. These results suggest that $\xi_i$ may serve as an indicator of the stability or variability of a team's pressing intensity.

\begin{figure}[h]
	\includegraphics[width=16cm]{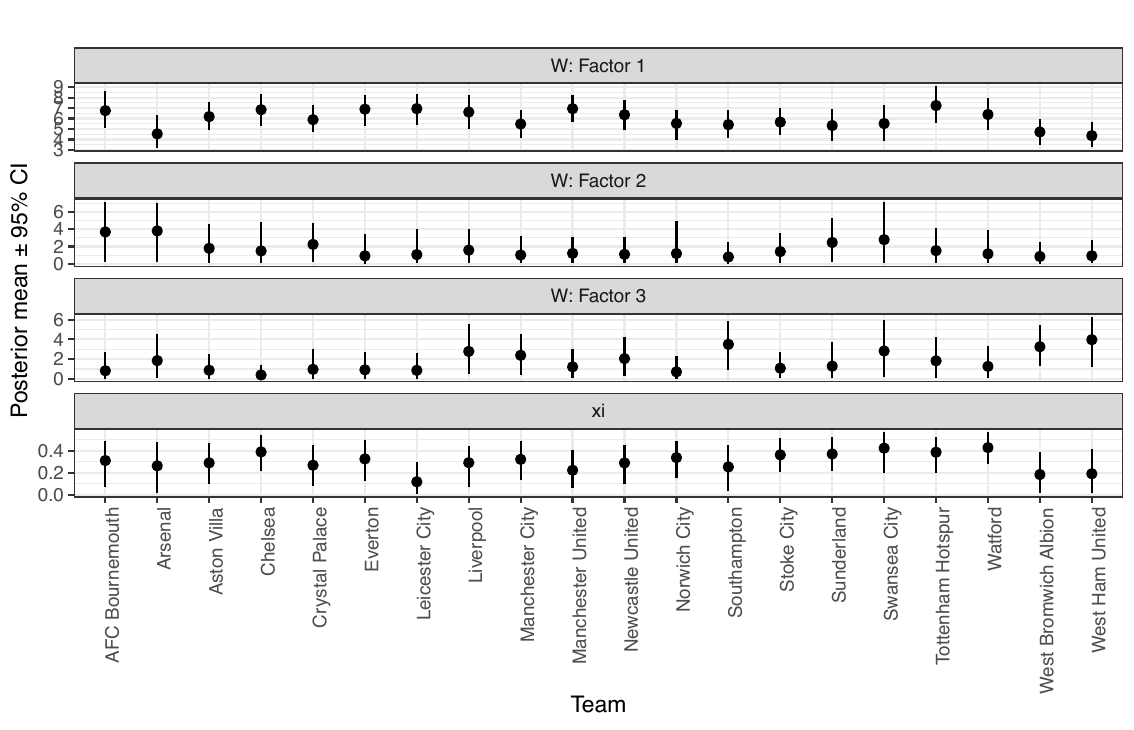}
	\caption{Posterior means and 95\% credible intervals for $W$ and $\boldsymbol{\xi}$ obtained by Bayesian GPMF.}
	\label{fig:w_bgpmf}
\end{figure}

\begin{figure}[h]
	\includegraphics[width=16cm]{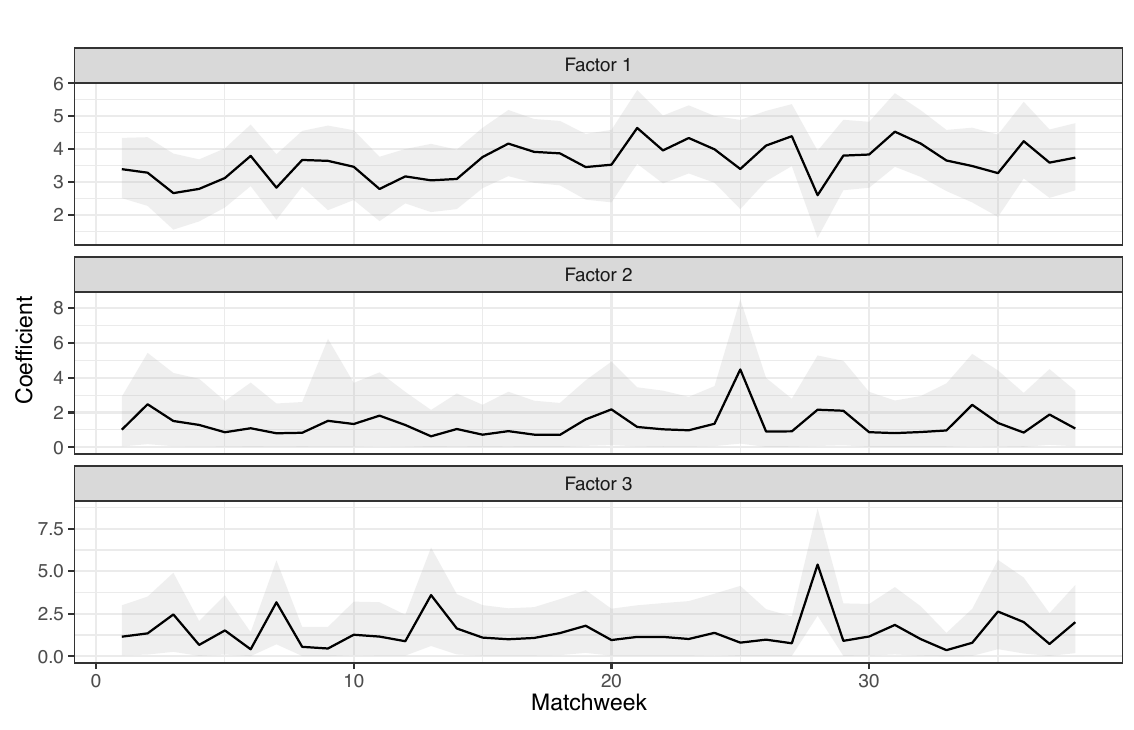}
	\caption{Posterior means and 95\% credible intervals for $H$ obtained by Bayesian GPMF.}
	\label{fig:h_bgpmf}
\end{figure}

\section{Conclusion}

In this paper, we proposed Bayesian GPMF, a Bayesian extension of generalized Poisson matrix factorization (GPMF). By introducing gamma and exponentially tilted beta distributions as prior distributions, we developed a Gibbs sampler for posterior inference.

For sampling from the exponentially tilted beta distribution in the Gibbs sampler, we also derived an infinite mixture representation of the EBeta distribution and employed a finite approximation based on this representation. Simulation experiments showed that the finite approximation achieved estimation accuracy comparable to Sampling/Importance Resampling (SIR), while providing greater computational efficiency.

The simulation experiments also showed that Bayesian GPMF yielded smaller mean squared errors than the conventional GPMF, while the empirical coverage probabilities of the credible intervals were generally close to the nominal level. Furthermore, the application to real data demonstrated the usefulness of Bayesian GPMF.

Future work includes selection of the number of factors and the development of faster inference algorithms for larger datasets.

\printcredits

\bibliographystyle{cas-model2-names}

\bibliography{sakaori-bibliography}



\end{document}